\documentclass[conference]{IEEEtran}
\IEEEoverridecommandlockouts
\usepackage{cite}
\usepackage{amsmath,amssymb,amsfonts,amsthm}
\usepackage{algorithmic}
\usepackage{graphicx}
\usepackage{textcomp}
\usepackage{xcolor}
\usepackage{color}
\usepackage{algorithmic}
\usepackage{url}
\usepackage{lipsum}
\usepackage{multirow}
\usepackage{afterpage}
\usepackage{bbm}
\usepackage{dsfont}
\usepackage[ruled,lined,boxed]{algorithm2e}
\usepackage{mathtools, nccmath}{\tiny {\tiny }}
\usepackage{lipsum}
\usepackage{bm}
\usepackage{threeparttable}
\usepackage{svg}
\usepackage{bbm}
\usepackage{amsmath}
\usepackage{csquotes}
\usepackage[T1]{fontenc}
\usepackage{threeparttable}
\usepackage{multirow}
\usepackage{makecell}
\usepackage{tikz}
\usepackage{tikz-3dplot}
\usepackage{pgfplots}
\pgfplotsset{compat=1.15}
\usepackage{subcaption}
\usepackage{siunitx}

\DeclareMathAlphabet{\mathbcal}{OMS}{cmsy}{b}{n}
\def\BibTeX{{\rm B\kern-.05em{\sc i\kern-.025em b}\kern-.08em
		T\kern-.1667em\lower.7ex\hbox{E}\kern-.125emX}}
 
\newtheorem{prop}{Proposition}

\newtheoremstyle{iremark}
{\topsep}   
{\topsep}   
{\upshape}  
{0pt}       
{\itshape}  
{:}         
{5pt plus 1pt minus 1pt} 
{\thmname{#1}\thmnumber{ \itshape#2}\thmnote{ (#3)}} 
\theoremstyle{iremark}
\newtheorem{remark}{Remark}

\allowdisplaybreaks
\mathchardef\mhyphen="2D
\makeatletter 
\let\myorg@bibitem\bibitem
\def\bibitem#1#2\par{%
	\@ifundefined{bibitem@#1}{%
		\myorg@bibitem{#1}#2\par
	}{%
		\begingroup
		\color{\csname bibitem@#1\endcsname}%
		\myorg@bibitem{#1}#2\par
		\endgroup
	}%
}
\makeatother

\makeatletter
\def\@IEEEtitleabstractindextextfont{%
    \fontsize{10pt}{10pt}\bfseries\selectfont 
}
\makeatother

\begin{document}
{\title{A General EM-Based Channel Model for Reconfigurable Antenna Systems\\
}

\author{
  \IEEEauthorblockN{
   Chen~Xu~
   and Xianghao~Yu
  }
  \IEEEauthorblockA{
   Department of Electrical Engineering, City University of Hong Kong, Hong Kong}
  Email: cxu297-c@my.cityu.edu.hk, alex.yu@cityu.edu.hk
 }

\maketitle

\begin{abstract}
Reconfigurable antenna systems (RASs), such as fluid antennas and movable antennas, are poised to play a pivotal role in sixth-generation (6G) systems by dynamically adapting the antenna elements for system performance enhancement.
However, unlocking their full potential requires channel models that accurately capture the influence of antenna configurations on the radiation, propagation, and reception of signals.
Existing channel models suffer from several limitations, such as neglecting polarization effects, being restricted to specific antenna types, or relying on oversimplified assumptions.
In this paper, we propose a general electromagnetic (EM)-based channel model grounded in spherical vector wave expansion (SVWE).
The proposed EM-based channel model captures the impact of antenna position and orientation on the channel gain, thereby making it particularly well-suited for RASs.
The effectiveness and accuracy are validated through comparisons with commercial simulation software, demonstrating excellent agreement in predicted channel gains.
{\color{black}{Moreover, it is shown that antenna orientation is a critical factor governing communication performance, and that dynamically adjusting the antenna orientation yields up to 70\% improvement in achievable communication rate compared to a fixed-antenna configuration.}}
\end{abstract}



\section{Introduction}
As an emerging antenna technology, reconfigurable antenna systems (RASs), such as fluid antennas~\cite{9264694} and movable antennas~\cite{liu2025ma,10286328}, dynamically adjust antenna positions and orientations to optimize communication performance in real time, thereby holding significant promise for sixth-generation (6G) networks.
However, to fully unlock the potential of RASs, accurate channel models that precisely capture the electromagnetic (EM) interaction between the antennas and the propagation environment are essential.
Specifically, the channel model is expected to account for the vectorial nature of EM propagation, capture the radiation and reception characteristics of antennas, and be applicable to arbitrary reconfigurable antenna element types.

In response to these requirements, recent efforts have shifted toward EM-based channel modeling, which seeks to construct channel models from a fundamental EM perspective.
For instance, to obtain a physically consistent representation of fading channels, the authors of \cite{11006094} modeled small-scale fading as a random \emph{scalar} field.
However, this model neglects the \emph{vectorial} nature of EM fields, and it fails to fully capture key EM effects, such as polarization mismatch, thereby leading to inaccurate characterization of the channel gain in RASs.
In \cite{11258091}, building upon the closed-form expression of the vector radiation field, the authors derived an EM-based channel model for RASs. 
Although this model intrinsically captures vectorial EM effects, it is restricted to scenarios where both the transceivers employ half-wave dipoles, thereby significantly limiting its applicability to more general antenna configurations.
To enhance the generality of channel modeling, the authors of \cite{10500425} and \cite{10500751} established a channel model based on dyadic Green's functions and applied it to extremely large-scale multiple-input multiple-output (XL-MIMO) and holographic MIMO (HMIMO) systems, respectively.
In principle, the model is applicable to arbitrary antenna types.
However, it requires prior knowledge of the antenna's surface current distribution, which is inaccessible in practice, limiting its applicability in real-world scenarios.
Another theoretical framework, known as the spherical vector wave expansion (SVWE), can be employed to analyze the EM propagation between arbitrary antennas.
In contrast to Green's function-based methods, SVWE only relies on sparsely sampled values of the EM field, which is readily measurable in practical engineering setups, to analytically characterize the EM propagation between transmit and receive antennas.
Owing to this, SVWE has been widely adopted in engineering applications such as antenna measurements~\cite{Hansen_1988}.
The authors of \cite{5159555} first employed SVWE to channel modeling.
Nevertheless, for simplicity, the derived channel model neglected the impact of the transmit antenna's type, position, and orientation on the channel gain, making it incompatible with RASs.


In this paper, we develop a comprehensive EM-based channel model which intrinsically accounts for the impact of arbitrary antenna positions and orientations on the channel gain, making it particularly well-suited for the performance characterization and optimization of RASs.
In particular, by decomposing the signal propagation into three fundamental EM processes, we quantitatively analyze their characteristics within the SVWE framework.
The proposed channel model relies on practically reasonable assumptions and is applicable to arbitrary antenna types, thereby offering broad applicability and practical utility.
The accuracy and validity of the proposed model are evaluated through full-wave simulations, demonstrating excellent agreement in channel gain predictions.

\section{Analysis of Radiation, Propagation, and
Reception Processes}\label{sec:EM_char_analysis}

Consider an RAS-enabled MIMO system comprising  $N_{\mathrm{t}}$ transmit and $N_{\mathrm{r}}$ receive antennas.
Denote the transmit signal vector as $\mathbf{v} \in \mathbb{C}^{N_{\mathrm{t}} \times 1}$.
The corresponding receive signal vector $\mathbf{w} \in \mathbb{C}^{N_{\mathrm{t}} \times 1}$ is then given by
\begin{equation}
	\begin{aligned}
  		\mathbf{w} = \mathbf{H} \mathbf{v} + \mathbf{z},
  	\end{aligned}
\end{equation}
where $\mathbf{H} \in \mathbb{C}^{N_{\mathrm{r}} \times N_\mathrm{t}}$ is the channel matrix and $\mathbf{z} \in \mathbb{C}^{N_{\mathrm{r}} \times 1}$ represents the additional white Gaussian noise (AWGN) vector whose entries are independent and identically distributed (i.i.d.) complex Gaussian variables with zero mean and variance $\sigma_\mathrm{n}^2$.

\begin{figure}[t]
	\centering
	\hspace{4mm}\includegraphics[trim=10 1 0 11, clip, width=0.43\textwidth]{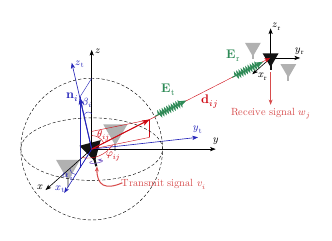}
	\caption{Illustration of the channel model between the $i$-th transmit antenna and the $j$-th receive antenna.}
	\label{fig:MIMO_setting_illu} 
	\vspace{-3mm}  
\end{figure}

As shown in Fig. \ref{fig:MIMO_setting_illu}, local coordinate systems $(x_\mathrm{t},y_\mathrm{t},z_\mathrm{t})$ and $(x_\mathrm{r},y_\mathrm{r},z_\mathrm{r})$ are established with their origins at the $i$-th transmit and $j$-th receive antennas, respectively, and their $z$-axes are aligned with the normal direction of the corresponding antennas.
Let $h_{ij}$ denote the EM-based channel gain from the $i$-th transmit antenna to the $j$-th receive antenna, i.e., $h_{ij}=[\mathbf{H}]_{i,j}$.
A global coordinate system $(x,y,z)$ is then defined with its origin at the transmit antenna and its $z$-axis parallel to the $z_\mathrm{r}$-axis.
In this global frame $(x,y,z)$, the normal direction vector $\mathbf{n}_i$ of the $i$-th transmit antenna and the displacement vector $\mathbf{d}_{ij}$ from the $i$-th transmit antenna to the $j$-th receive antenna are characterized by azimuth-elevation angle pairs $(\alpha_i,\beta_i)$ and $(\varphi_{ij},\theta_{ij})$, respectively.

In the remainder of this section, within the SVWE framework, we analyze the channel characteristics from the $i$-th transmit antenna to the $j$-th receive antenna from an EM propagation perspective. 
As illustrated in Fig. \ref{fig:MIMO_setting_illu}, the transformation from the transmit signal $v_i$ to the received signal $w_j$, which underlies the classical wireless channel model $h_{ij}\triangleq w_j/v_i$, can be divided by three fundamental EM processes:
 A) the conversion of the input signal $v_i$ into the radiated electric field $\mathbf{E}_{\mathrm{t}}$ by the $i$-th transmit antenna;
 B) the propagation of $\mathbf{E}_{\mathrm{t}}$ through free space, during which the portion impinging on the receive antenna becomes the incident electric field $\mathbf{E}_{\mathrm{r}}$;
 and C) the conversion of the incident field $\mathbf{E}_{\mathrm{r}}$ into the output signal $w_j$ at the $j$-th receive port. 
 These processes correspond to the EM characterizations of radiation, propagation, and reception, which will be conducted in the following three subsections.

\subsection{Radiation Characterization}

The vector Helmholtz equation governs any time-harmonic EM field in source-free regions. 
As its general solution in spherical coordinates, the spherical vector wave functions (SVWFs) constitute a complete and orthogonal basis for representing arbitrary time-harmonic EM fields. 
By expanding the field as a linear superposition of SVWFs, the SVWE provides a mathematically complete and physically exact framework for electromagnetic field characterization \cite{Hansen_1988}.
Unlike plane- and cylindrical-wave expansions, SVWE leverages the inherent symmetry of SVWFs, and thereby enables a natural and efficient characterization of how antenna rotations and translations shape the EM field.
This makes SVWE particularly well-suited for analyzing channel characteristics in RAS-enabled systems.

According to the SVWE framework, when observed in the local spherical coordinate system $(r_\mathrm{t},\theta_\mathrm{t},\varphi_\mathrm{t})$ associated with $(x_\mathrm{t},y_\mathrm{t},z_\mathrm{t})$, the radiated electric field $\mathbf{E}_{\mathrm{t}}$ can be decomposed as
\begin{equation}\label{eq:elec_field_as_sum_spherical_func}
	\begin{aligned}
  		\mathbf{E}_{\mathrm{t}}(r_\mathrm{t},\theta_\mathrm{t},\varphi_\mathrm{t})=\frac{k v_i}{\sqrt{\eta}}  \sum_{s=1}^{2}\sum_{n=1}^{N}\sum_{m=-n}^{n}   T_{smn}\mathbf{F}_{smn}^ {(3)} (r_\mathrm{t},\theta_\mathrm{t},\varphi_\mathrm{t}),
  	\end{aligned}
\end{equation}
where $k$ and $\eta$ denote the wavenumber and the wave impedance of the medium, respectively. 
The truncation index $N$ is a constant determined by the electrical size of the antenna, typically chosen as $N = \lceil ka \rceil$, where $a$ is the radius of the smallest sphere enclosing the antenna.
In (\ref{eq:elec_field_as_sum_spherical_func}), $\mathbf{F}^{(3)}_{smn}$ represents the outward-propagation spherical vector wave function (SVWF) of mode $(s,m,n)$ and $T_{smn}$ is the corresponding expansion coefficient.
The explicit expressions of SVWFs are provided in Appendix \ref{sec:appendix_intro_svwfs}.

For a specific SVWF $\mathbf{F}^{(c)}_{smn}$, the index $c \in \{ 1,2,3,4\}$, where $c=1$ or $2$ corresponds to standing waves, and $c=3$ and $4$ represent outward- and inward-propagation traveling waves, respectively.
The index $s$ distinguishes the wave polarization type, with $s=1$ denoting transverse electric (TE) modes and $s=2$ denoting transverse magnetic (TM) modes.
The indices $n$ and $m$ characterize the angular dependence of the SVWFs.
Specifically, $n$ determines the polar variation and overall angular complexity, while 
$m$ governs the azimuthal periodicity. 
The explicit formulations of SVWFs $\mathbf{F}^{(c)}_{smn}$ are provided in \cite[eq.~(2.20), (2.21)]{Hansen_1988}.

The coefficient $T_{smn}$ represents the antenna's ability to convert an excitation signal $v_i$ into the corresponding spherical vector wave SVWF of mode index $(s,m,n)$,  also referred to as the radiation coefficient. 
It is fully determined by the antenna's physical properties, e.g., materials, structure, and feeding network, and therefore $T_{smn}$ are typically known a priori for channel modeling.

\begin{remark}
It is worth noting that conventional Green's function-based approaches\cite{10500425,10500751} require dynamic prior knowledge of the antenna's surface current distribution, which is typically inaccessible in real-world engineering scenarios.
In contrast, SVWE only relies on static and measurable radiation coefficients $T_{smn}$ to reconstruct the antenna's radiated electric field, making it particularly well-suited for characterizing channel models in practical communication systems.
\end{remark}

\subsection{Propagation Characterization}
In an RAS-enabled system, variation in the orientation of the transmit antenna alters the radiation field $\mathbf{E}_{\mathrm{t}}$, which in turn significantly affects the incident field $\mathbf{E}_{\mathrm{r}}$ at the receive antenna. 
Consequently, the propagation characteristics from $\mathbf{E}_{\mathrm{t}}$ to $\mathbf{E}_{\mathrm{r}}$ must be carefully analyzed to accurately model the impact of transmitter orientation variations on channel gain.
Note that in (\ref{eq:elec_field_as_sum_spherical_func}), we obtain the expression for the radiated electric field $\mathbf{E}_{\mathrm{t}}$ in the coordinate system $(r_\mathrm{t},\theta_\mathrm{t},\varphi_\mathrm{t})$.
To quantitatively analyze the transformation from $\mathbf{E}_{\mathrm{t}}$ to $\mathbf{E}_{\mathrm{r}}$,
we first transform $\mathbf{E}_{\mathrm{t}}$ into the spherical coordinate system $(r_\mathrm{r},\theta_\mathrm{r},\varphi_\mathrm{r})$ associated with $(x_\mathrm{r},y_\mathrm{r},z_\mathrm{r})$, and then isolate the inward-propagating component, which constitutes the incident field $\mathbf{E}_{\mathrm{r}}$ at the receiver.
In fact, owing to the rotation and translation properties of SVWFs, $\mathbf{E}_{\mathrm{t}}$ can be rigorously transformed from coordinate system $(r_\mathrm{t},\theta_\mathrm{t},\varphi_\mathrm{t})$ to $(r_\mathrm{r},\theta_\mathrm{r},\varphi_\mathrm{r})$~\cite{Hansen_1988}.

Specifically, consider a coordinate system $(x,y,z)$ that is first rotated by an angle $\varphi$ about the $z$-axis and then by an angle $\theta$ about the new $y$-axis, resulting in the rotated system $(x',y',z')$. The rotation property of SVWF can be expressed as
\begin{equation}\label{eq:rotation_formula}
	\begin{aligned}
  		\mathbf{F}^{(c)}_{smn}(r,\theta,\varphi) = \sum_{\mu=-n}^n  e^{\jmath m \varphi} d^n_{\mu m}(\theta)  \mathbf{F}^{(c)}_{s \mu n}(r',\theta',\varphi'),
  	\end{aligned}
\end{equation}
where the rotation coefficient $d^n_{\mu m}(\theta)$ is a complex-valued function, and its explicit expression is provided in \cite[eq.~(A2.5)]{Hansen_1988}.

If the coordinate system $(x,y,z)$ is translated by a distance $A$ along the $z$-axis to obtain $(x',y',z')$, the translation property of SVWF can be expressed as
\begin{equation}\label{eq:translation_formula}
	\begin{aligned}
  		\mathbf{F}_{{s \mu n}}^{(c)}(r, \theta, \varphi) = 
								\sum_{\sigma=1}^{2} 
								\sum_{\substack{\nu=|\mu| \\ \nu \neq 0}}^{N} 
								C_{\sigma \mu \nu}^{{sn (c)}}(k A) 
								\mathbf{F}_{\sigma \mu \nu}^{(1)}(r', \theta', \varphi')
  	\end{aligned}
\end{equation}
for $r' < |A|$.
The translation coefficient $C_{\sigma \mu \nu}^{{sn(c)}}(k A)$ is a complex-valued function, and its explicit expression is provided in \cite[eq.~(A3.3)]{Hansen_1988}.

The analytical expressions (\ref{eq:rotation_formula}) and (\ref{eq:translation_formula}) for the rotation and translation of SVWFs fundamentally convey a key principle that a single-mode spherical wave, when subjected to a coordinate transformation, becomes a superposition of multiple spherical modes in the new coordinate system.
By successively applying (\ref{eq:rotation_formula}) and (\ref{eq:translation_formula}), the SVWF $\mathbf{F}_{{s m n}}^{(c)}(r_\mathrm{t}, \theta_\mathrm{t}, \varphi_\mathrm{t})$ can be expressed in an arbitrarily rotated and translated coordinate system $(r_\mathrm{r}, \theta_\mathrm{r}, \varphi_\mathrm{r})$ as $\mathbf{F}_{{s m n}}^{(c)}(r_\mathrm{r}, \theta_\mathrm{r}, \varphi_\mathrm{r})$.
Substituting $\mathbf{F}_{{s m n}}^{(c)}(r_\mathrm{r}, \theta_\mathrm{r}, \varphi_\mathrm{r})$ into (\ref{eq:elec_field_as_sum_spherical_func}) yields the SVWE of the radiated field $\mathbf{E}_{\mathrm{t}}$ in the coordinate system $(r_\mathrm{r}, \theta_\mathrm{r}, \varphi_\mathrm{r})$, namely $\mathbf{E}_{\mathrm{t}}(r_\mathrm{r}, \theta_\mathrm{r}, \varphi_\mathrm{r})$.

Once the SVWE of $\mathbf{E}_{\mathrm{t}}(r_\mathrm{r}, \theta_\mathrm{r}, \varphi_\mathrm{r})$ is derived, by retaining the inward-propagating modes, i.e., $c=4$, we obtain the SVWE of the incident field $\mathbf{E}_{\mathrm{r}}$ at the receive antenna.
This procedure provides a complete electromagnetic description of the propagation from the radiation field $\mathbf{E}_{\mathrm{t}}$ to the incident field $\mathbf{E}_{\mathrm{r}}$.

\subsection{Reception Characterization}
Suppose that the incident electric field $\mathbf{E}_{\mathrm{r}}$ at the receive antenna admits an SVWE in the form of
\begin{equation}\label{eq:Er}
	\begin{aligned}
  		\mathbf{E}_{\mathrm{r}}(r_\mathrm{r},\theta_\mathrm{r},\varphi_\mathrm{r})= \sum_{s=1}^{2}\sum_{n=1}^{N}\sum_{m=-n}^{n}   Q_{smn}\mathbf{F}_{smn}^ {(4)} (r_\mathrm{r},\theta_\mathrm{r},\varphi_\mathrm{r}).
  	\end{aligned}
\end{equation}
Then, the received signal $w_j$ induced by $\mathbf{E}_{\mathrm{r}}$ at the $j$-th receive antenna is expressed as 
\begin{equation}\label{eq:Er_receive_signal_form}
	\begin{aligned}
  		w_j = \sum_{s=1}^{2}\sum_{n=1}^{N}\sum_{m=-n}^{n} Q_{smn} R_{smn},
  	\end{aligned}
\end{equation}
where $R_{smn}$ is the reception coefficient, characterizing the antenna's ability to convert an incident spherical mode $(s,m,n)$ into the receive signal at ports.

For reciprocal antennas, which constitute the vast majority of practical radiating elements, the reception and radiation characteristics are related by electromagnetic reciprocity.
Specifically, the reception coefficient $R_{smn}$ can be directly determined from the radiation coefficient $T_{smn}$~\cite{Hansen_1988}, namely
\begin{equation}
	\begin{aligned}
  		R_{smn} = (-1)^m T_{s,-m,n}.
  	\end{aligned}
\end{equation}

Building upon the above analysis of the three fundamental physical processes, i.e., radiation, propagation, and reception, that determine the EM characteristics of the wireless channel, we have established a foundational understanding of the EM-based channel model. 
This provides the theoretical foundation for developing a channel model tailored to RAS-enabled MIMO systems.

\section{EM-Based Channel Model for RAS-Enabled MIMO}

Building on the SVWE analysis in Sec. \ref{sec:EM_char_analysis}, we then derive the explicit expression for the EM-based channel model $h_ij$ of the RAS-enabled MIMO system.
To enable a compact matrix formulation of the derived EM-based channel gain, we introduce a one-to-one index mapping $(s,m,n) \leftrightarrow  j$ that compresses the multi-dimensional SVWF mode indices into a single index \cite{Hansen_1988}, and denote the total number of spherical vector wave modes as $J$.
By stacking the radiation coefficients $T_{smn}$ and receive coefficients $R_{smn}$ into the radiation coefficient vector $\mathbf{t} \in \mathbb{C}^{J \times 1}$ and reception coefficient vector $\mathbf{r} \in \mathbb{C}^{J \times 1}$, respectively, the following proposition provides an explicit expression for the $h_{ij}$.

\begin{prop}\label{prop:hij}
	The EM-based channel gain $h_{ij}$ is given by
\begin{equation}\label{eq:EM_based_SISO_channel_model}
	\begin{aligned}
  		h_{ij} = \frac{e^{\jmath k \Vert \mathbf{d}_{ij} \Vert}}{{k \Vert \mathbf{d}_{ij} \Vert}} \mathbf{r}^T \mathbf{G}(\mathbf{n}_{i},\mathbf{d}_{ij}) \mathbf{t},
  	\end{aligned}
\end{equation}
where $\jmath = \sqrt{-1}$. By leveraging the index mapping $(s,m,n) \leftrightarrow p$ and $(\sigma,\rho,\nu) \leftrightarrow q$, the $p,q$-th element of the matrix $\mathbf{G}(\mathbf{n}_{i},\mathbf{d}_{ij}) \in \mathbb{C}^{J \times J}$ is given by
\begin{equation}\label{eq:G_pq}
	\begin{aligned}
  		[{\mathbf{G}}(\mathbf{n}_{i},\mathbf{d}_{ij})]_{p,q} & =  \frac{k}{2\sqrt{\eta}} a(n,\nu)   e^{-\jmath \rho \varphi_{ij}} \\
  		&\sum_{\mu=-n}^{n} e^{\jmath \mu ( \varphi_{ij} - \alpha_i)}  
  		   b^{s,\mu,n}_{\sigma,\rho,\nu}(\theta_{ij}) d^{n}_{\mu m}(-\beta_{i}).
	\end{aligned}
\end{equation}
where
\begin{equation}\label{eq:anv}
	\begin{aligned}
  		a(n,\nu) = &\frac{\jmath ^{\nu-n-1}\sqrt{(2n+1)(2\nu+1)}}{2},
  	\end{aligned}
\end{equation}

\begin{equation}
	\begin{aligned}\label{eq:b_theta}
  		b^{s,\mu,n}_{\sigma,\rho,\nu}(\theta_{ij}) = (-1)^{s+\sigma}  d^{n}_{-1,\mu}(\theta_{ij}) &d^{\nu}_{-1,\rho}(\theta_{ij}) + \\ 
  		 &  d^{n}_{1,\mu}(\theta_{ij}) d^{\nu}_{1,\rho}(\theta_{ij}).
  	\end{aligned}
\end{equation}

\end{prop}
\begin{IEEEproof}
	Please refer to Appendix \ref{sec:appendix_prof_prop_hij}.
\end{IEEEproof}}

\begin{remark}
It can be observed that the EM-based channel gain given in  (\ref{eq:EM_based_SISO_channel_model}) is expressed as the multiplication of two terms. The first fractional term corresponds to a uniform spherical wave which depends solely on the distance between the transmit and receive antennas. 
The second term $\mathbf{r}^T \mathbf{G}(\mathbf{n}_{i},\mathbf{d}_{ij}) \mathbf{t}$ is a composite factor that incorporates the specific antenna types (reflected by $\mathbf{t}$ and $\mathbf{r}$), the orientation of the transmit antenna ($\alpha_i$ and $\beta_i$), and the relative position between the transmit and receive antennas ($\varphi_{ij}$ and $\theta_{ij}$).
In classical channel modeling for wireless communication systems, only the scalar spherical wave  in the first term is considered. 
These prevent accurate characterizations of how the antenna's position and orientation affect the channel gain.
In contrast, by leveraging the SVWE framework, the proposed EM-based channel model intrinsically accounts for the vectorial nature of EM propagation  in the second term.
Hence, the derived channel model is more effective in characterizing antenna placement and orientation in RAS-enabled MIMO systems. 
\end{remark}

\section{Simulation Results}

{\color{black}{In this section, we validate the accuracy and effectiveness of the proposed channel model, and evaluate the impact of the transmit antenna's position and orientation on the achievable rate under this channel model.}}
We consider a half-wavelength dipole antenna as a setup example at both transmit and receive antennas. 
The antenna's operating wavelength $\lambda=0.1$ m. 
{\color{black}{In the global coordinate system $(x,y,z)$, unless otherwise specified, the transmit antenna is placed at the origin $\mathcal{O}$, and the receive antenna is located at $\mathcal{P}=(8\lambda,10\lambda,8\lambda)$.
The orientation of the transmit antenna is specified by the azimuth angle $\alpha$ and the elevation angle $\beta$.}}
The transmit power $P_\mathrm{t} = 10$ dBm and the noise power $\sigma_{\mathrm{n}}^2 = -20 $ dBm.

\subsection{Validation of the EM-based Channel Model}

{\color{black}{We first validate the capability of the proposed EM-based channel model in capturing the impact of the antenna's orientation on the channel gain.
As shown in Fig. \ref{fig:channel_gain_svwe}, we compute the variation of the point-to-point channel gain between the transmit and receive antennas for different values of $\alpha$ and $\beta$ based on (\ref{eq:EM_based_SISO_channel_model}).
For comparison, we employ the finite element method (FEM), a widely used full-wave EM simulation technique, to calculate the corresponding channel gain at every $\ang{10}$ increment in both $\alpha$ and $\beta$, as depicted in Fig.~\ref{fig:channel_gain_hfss}.
The normalized mean square error (NMSE) between the two results is approximately $-30.75$ dB.}}
\begin{figure}[t]
    \centering
    \hspace{-11pt}
    \begin{subfigure}[t]{0.47\linewidth}
        \centering
        \vspace{-5pt}
        \includegraphics[trim=30 20 21 25, clip, width=\textwidth]{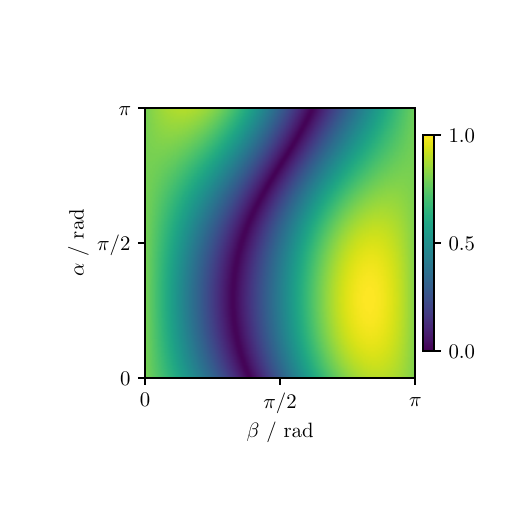}
        \caption{Derived channel model.}
        \label{fig:channel_gain_svwe}
    \end{subfigure}
   \hskip 0.04\columnwidth
    \begin{subfigure}[t]{0.47\linewidth}
        \centering
        \vspace{-5pt}
        \includegraphics[trim=30 20 21 25, clip, width=\textwidth]{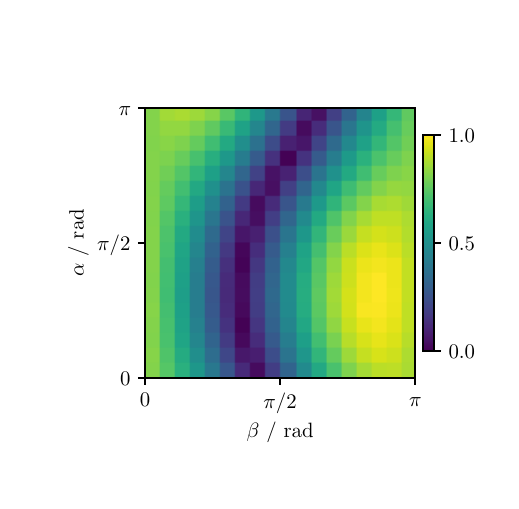}
        \caption{FEM full-wave simulation.}
        \label{fig:channel_gain_hfss}
    \end{subfigure}
    \caption{\color{black}{The variation of the normalized channel gain with respect to the azimuth angle $\alpha$ and elevation angle $\beta$ of the transmit antenna.}}
    \label{fig:channel_gain}
    \vspace{-1pt}
\end{figure}

{\color{black}{We then extend the validation to a typical $16 \times 16$ RAS-enabled MIMO system, where both the transmitter and receiver are equipped with half-wavelength spacing uniform linear arrays (ULAs).
Without loss of generality, we assume that both the transmit and receive antenna arrays are aligned along the positive $x$-axis, with their first elements located at $\mathcal{O}$ and $\mathcal{P}$, respectively.
Furthermore, to highlight the reconfigurability of the transmit array, we set the elevation and azimuth angles of the $i$-th transmit antenna to $\beta_i = 5i^{\circ}$ and $\alpha_i = \ang{90}$, respectively.
Fig. \ref{fig:MIMOchannel_gain_svwe} shows the computed EM-based channel matrix for the considered RAS-enabled MIMO system.
The corresponding FEM-based simulation result, obtained by computing the channel gain between every pair of transmit and receive antennas, is shown in Fig. \ref{fig:MIMOchannel_gain_hfss}.
The NMSE between the two results is approximately $-24.13$ dB.}}
\begin{figure}[t]
    \centering
    \hspace{-6.5pt}
    \begin{subfigure}[b]{0.47\linewidth}
        \centering
        \vspace{-14pt}
        \includegraphics[trim=30 30 29 25, clip, width=\textwidth]{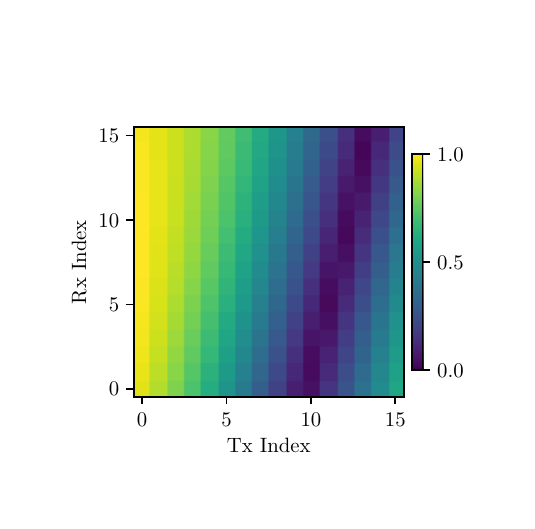}
        \caption{Derived channel model.}
        \label{fig:MIMOchannel_gain_svwe}
    \end{subfigure}
	\hskip 0.04\columnwidth
    \begin{subfigure}[b]{0.47\linewidth}
        \centering
        \vspace{-14pt}
        \includegraphics[trim=30 30 29 25, clip, width=\textwidth]{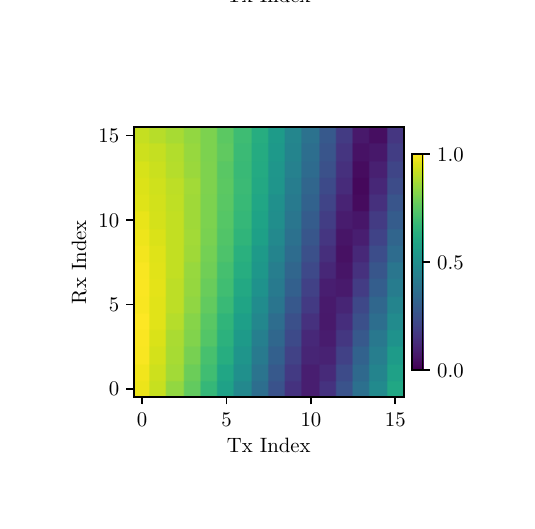}
        \caption{FEM full-wave simulation.}
        \label{fig:MIMOchannel_gain_hfss}
    \end{subfigure}
    \caption{\color{black}{The normalized channel gain of the considered $16 \times 16$ RAS-enabled MIMO system.}}
    \label{fig:MIMOchannel_gain}
\end{figure}

Both Fig.~\ref{fig:channel_gain} and Fig.~\ref{fig:MIMOchannel_gain} show excellent agreement between the analytically derived channel gain and the FEM-based simulation results, thereby clearly demonstrating that the proposed EM-based model accurately characterizes the channel characteristics of RAS-enabled MIMO systems.

\subsection{Impact of Transmit Antenna Position and Orientation on Achievable Rate}

{\color{black}{To separately investigate the impact of the transmit antenna's orientation and position on the achievable communication rate, defined by $\log \left( 1 + \vert h_{ij} \vert^2\ / \sigma_{\mathrm{n}}^2 \right)$, we conduct the following two experiments:
1) the antenna's orientation is fixed at $\alpha = \ang{0}$, $\beta = \ang{0}$, while its position is varied along the $x$-axis over a large interval $[-30\lambda,30\lambda]$;
2) the transmit antenna's position is fixed at $\mathcal{O}$ and its azimuth angle is set to $\alpha=\ang{20}$, while the elevation angle $\beta$ is varied over $[0,\pi]$.

\begin{figure}[t]
	\centering
	\setlength{\abovecaptionskip}{4pt}
	\vspace{10pt}
	\hspace{-9.99mm}
	\hspace{9mm}
	\includegraphics[trim=10 5 10 25, clip, width=0.406\textwidth]{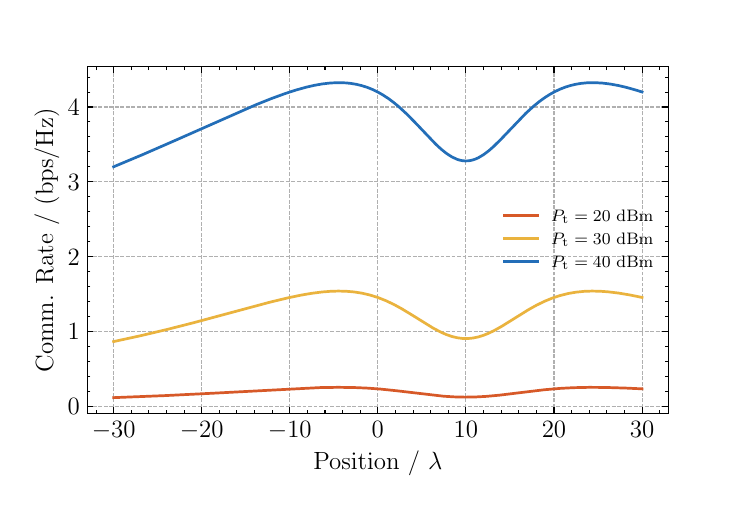}
	\caption{The achievable communication rate versus transmit antenna position for various transmit power.}
	\label{fig:rate_vs_pos}
	\vspace{2mm}  
\end{figure}

\begin{figure}[t]
	\setlength{\abovecaptionskip}{8pt}
	\centering
	\hspace{-14mm}
	\hspace{9mm}
	\includegraphics[trim=31 5 10 25, clip, width=0.4\textwidth]{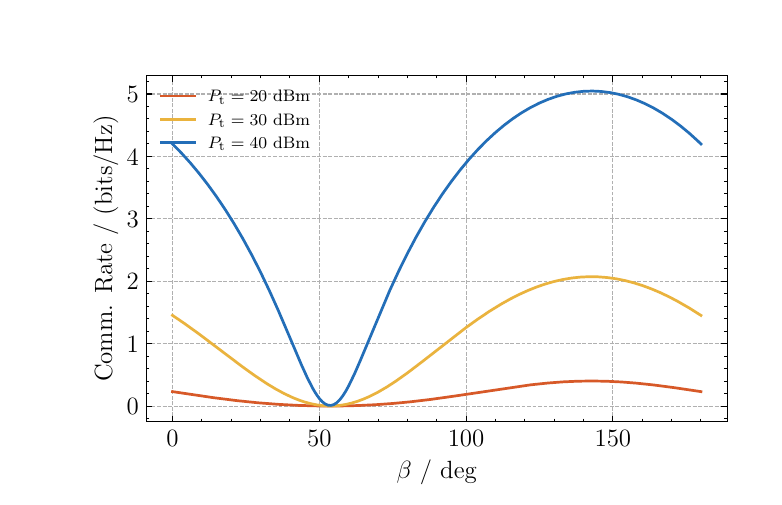}
	\caption{The achievable communication rate versus transmit antenna orientation for various transmit power.}
	\label{fig:rate_vs_beta}
	\vspace{0mm}  
\end{figure}

{\color{black}{As observed in Fig.~\ref{fig:rate_vs_pos}, compared to the communication rate achieved when the 
transmit antenna is located at $\mathcal{O}$, moving the antenna yields at most a $10.3$\% rate improvement at transmit power of $20$ dBm.
Moreover, as the transmit power increases to $30$ dBm and $40$ dBm, this gain decreases to $6.9$\% and $3.5$\%, respectively.
In contrast, in Fig.~\ref{fig:rate_vs_beta}, when taking the configuration with elevation angle $\beta = \ang{0}$ (typically representative of a fixed antenna deployment) as the baseline, adjusting $\beta$ improves the communication rate by approximately $73$\%, $42$\%, and $20$\% at the same transmit power levels.}}
{\color{black}{These results reveal that compared to the antenna position, the antenna orientation has a more pronounced impact on the achievable communication rate.
This finding is also attributed to the proposed EM-based channel model that accurately captures the impact of antenna orientation on channel gain in RAS-enabled communication systems.}}


\section{Conclusion}
This paper presented an EM-based channel model tailored for RAS-enabled MIMO systems.
By leveraging the SVWE of the EM field, we characterized the EM transmission between transmit and receive antennas, leading to a channel model that intrinsically incorporates the vectorial nature of EM fields, enabling accurate representation of polarization mismatch effects on the channel gain.
Numerical simulations validate the effectiveness and accuracy of the proposed model and further reveal the significant impact of transmit antenna orientation on achievable communication rates, thereby laying a foundational framework for the design and optimization of high-performance 6G communication systems based on RASs.

\appendices

\section{Introduction of SVWFs}\label{sec:appendix_intro_svwfs}
{\color{black}{SVWFs constitute a set of solutions to the vector Helmholtz equation expressed in spherical coordinates, and form a complete and orthogonal basis for representing electromagnetic field distributions. The explicit expressions of SVWFs are given by~\cite{Hansen_1988}
\begin{equation}\label{eq:F1mn}
	\begin{split}
		\mathbf{F}^{(c)}_{1mn}(r,\theta,\varphi) = &\frac{1}{\sqrt{2\pi}}\frac{1}{\sqrt{n(n+1)}}\left(-\frac{m}{|m|}\right)^{m} \\ 
		&\biggl\{z_{n}^{(c)}(k r)\frac{\jmath m\bar{P}_{n}^{|m|}(\cos\theta)}{\sin\theta}\mathrm{e}^{\jmath m\varphi}\hat{\theta}\\
  		&-z_{n}^{(c)}(k r)\frac{\mathrm{d}\bar{P}_{n}^{|m|}(\cos\theta)}{\mathrm{d}\theta}\mathrm{e}^{\jmath m\varphi}\hat{\varphi}\biggr\}
	\end{split}
\end{equation}
and
\begin{equation}\label{eq:F2mn}
	\begin{split}
		\mathbf{F}^{(c)}_{2mn} & (r,\theta,\varphi) = \frac{1}{\sqrt{2\pi}}\frac{1}{\sqrt{n(n+1)}}\left(-\frac{m}{|m|}\right)^{m} \\
& \biggl\{ \frac{n(n+1)}{k r}z_{n}^{(c)}(kr)\bar{P}_{n}^{|m|}(\cos\theta)\mathrm{e}^{\jmath m\varphi}\hat{r}\\
&+\frac{1}{k r}\frac{\mathrm{d}}{\mathrm{d}(k r)}\left\{k rz_{n}^{(c)}(kr)\right\}\frac{\mathrm{d}\bar{P}_{n}^{|m|}(\cos\theta)}{\mathrm{d}\theta}\mathrm{e}^{\jmath m\varphi}\hat{\theta}\\
&+\frac{1}{k r}\frac{\mathrm{d}}{\mathrm{d}(k r)}\left\{k rz_{n}^{(c)}(kr)\right\}\frac{\jmath m\bar{P}_{n}^{|m|}(\cos\theta)}{\sin\theta}\mathrm{e}^{\jmath m\varphi}\hat{\varphi} \biggr\},
	\end{split}
\end{equation}
where $\bar{P}^m_n(\cos \theta)$ is the normalized associated Legendre function and $z_{n}^{(c)}(kr)$ is radial function which is specified by index $c$ as
\begin{subequations}\label{eq:zn}
	\begin{align}
  		&z_n^{(1)}(kr)	=	j_n(kr), \\ 
  		&z_n^{(2)}(kr)	=	n_n(kr), \\ 
  		&z_n^{(3)}(kr)	=	j_n(kr) + \jmath n_n(kr), \\ 
  		&z_n^{(4)}(kr)	=	j_n(kr) - \jmath n_n(kr).
  	\end{align}
\end{subequations}
$j_n(kr)$ and $n_n(kr)$ are spherical Bessel function and spherical Neumann function, respectively.
}}

\section{Proof of Proposition \ref{prop:hij}}\label{sec:appendix_prof_prop_hij}

We first examine the transformation relationship of a single SVWF $\mathbf{F}_{smn}^{(c)}$ between the two coordinate systems $(x_\mathrm{t},y_\mathrm{t},z_\mathrm{t})$ and $(x_\mathrm{r},y_\mathrm{r},z_\mathrm{r})$ depicted in Fig. \ref{fig:MIMO_setting_illu}.
From a spatial perspective, the transformation from  $(x_\mathrm{t},y_\mathrm{t},z_\mathrm{t})$ to $(x_\mathrm{r},y_\mathrm{r},z_\mathrm{r})$ can be realized through the following sequence of operations:
\begin{itemize}
	\item Rotate $(x_\mathrm{t},y_\mathrm{t},z_\mathrm{t})$ by an angle of $-\beta_{i}$ about the $y_\mathrm{t}$-axis to obtain $(x_1,y_1,z_1)$;
	\item Rotate $(x_1,y_1,z_1)$ by an angle of $-\alpha_{i}$ about the $y_1$-axis to obtain $(x_2,y_2,z_2)$;
	\item Rotate $(x_2,y_2,z_2)$ by an angle of $\varphi_{ij}$ about the $z_2$-axis to obtain $(x_3,y_3,z_3)$;
	\item Rotate $(x_3,y_3,z_3)$ by an angle of $\theta_{ij}$ about the $y_3$-axis to obtain $(x_4,y_4,z_4)$;
	\item Translate $(x_4,y_4,z_4)$ by a distance $\Vert \mathbf{d}_{ij} \Vert$ along the $z_4$-axis to obtain $(x_5,y_5,z_5)$;;
	\item Rotate $(x_5,y_5,z_5)$ by an angle of $-\theta_{ij}$ about the $y_5$-axis to obtain $(x_6,y_6,z_6)$;
	\item Rotate $(x_6,y_6,z_6)$ by an angle of $-\varphi_{ij}$ about the $z_6$-axis to arrive at the target coordinate system $(x_\mathrm{r},y_\mathrm{r},z_\mathrm{r})$;
\end{itemize}
The SVWF transformation corresponding to each of the above steps is given by (\ref{eq:rotation_formula}) or (\ref{eq:translation_formula}).
Specifically, for an SVWF $\mathbf{F}_{smn}^{(c)}(r_\mathrm{t},\theta_\mathrm{t},\varphi_\mathrm{t})$ defined in the coordinate system $(x_\mathrm{t},y_\mathrm{t},z_\mathrm{t})$, its representation in the coordinate system $(x_\mathrm{r},y_\mathrm{r},z_\mathrm{r})$ is given by
\begin{equation}\label{eq:Fsmn_general_transform}
	\begin{aligned}
  		\mathbf{F}_{smn}^{(c)}  (r_\mathrm{t},\theta_\mathrm{t},\varphi_\mathrm{t}) = \sum_{\mu=-n}^{n}  \sum_{\gamma=-n}^{n} \sum_{\sigma=1}^{2}  \sum_{\substack{\nu=|\gamma| \\ \nu \neq 0}}^{N} \sum_{\rho=-\nu}^{\nu}
  		e^{\jmath [(\mu - \rho) \varphi_{ij} -\mu  \alpha_i]} \\
  		d^n_{\mu m}(-\beta_{i})  d^{n}_{\gamma \mu}(\theta_{ij}) d^{\nu}_{\rho \gamma}(-\theta_{ij})  C^{sn(c)}_{\sigma \gamma \nu}(k \Vert \mathbf{d}_{ij} \Vert) \mathbf{F}_{\sigma \rho \nu}^{(1)}(r_\mathrm{r},\theta_\mathrm{r},\varphi_\mathrm{r}).
  	\end{aligned}
\end{equation}
Insert (\ref{eq:Fsmn_general_transform}) into (\ref{eq:elec_field_as_sum_spherical_func}), we obtain
\begin{equation}\label{eq:Et_in_Or}
	\begin{aligned}
  		\mathbf{E}_{\mathrm{t}} & (r_\mathrm{r},\theta_\mathrm{r},\varphi_\mathrm{r}) = 
  		\frac{k v_i}{\sqrt{\eta}} \sum_{s=1}^{2} \sum_{n=1}^{N}\sum_{m=-n}^{n}  \sum_{\mu=-n}^{n}  \sum_{\gamma=-n}^{n} \sum_{\sigma=1}^{2}  \sum_{\substack{\nu=|\gamma| \\ \nu \neq 0}}^{N} \sum_{\rho=-\nu}^{\nu} \\
  		 & T_{smn} e^{\jmath [(\mu - \rho) \varphi_{ij} -\mu  \alpha_i]}  d^n_{\mu m}(-\beta_{i})  d^{n}_{\gamma \mu}(\theta_{ij}) d^{\nu}_{\rho \gamma}(-\theta_{ij})   \\
  		 & C^{sn(c)}_{\sigma \gamma \nu}(k \Vert \mathbf{d}_{ij} \Vert)  \mathbf{F}_{\sigma \rho \nu}^{(1)}(r_\mathrm{r},\theta_\mathrm{r},\varphi_\mathrm{r}).
  	\end{aligned}
\end{equation}
(\ref{eq:Et_in_Or}) indicates that, in the coordinate system $(x_\mathrm{r},y_\mathrm{r},z_\mathrm{r})$, the radiated field of the transmit antenna can be represented as a superposition of infinitely many standing waves.
Each standing wave can be decomposed into two counter-propagating traveling waves, i.e.,
\begin{equation}
	\begin{aligned}
  		\mathbf{F}_{\sigma \rho \nu}^{(1)}(r_\mathrm{r},\theta_\mathrm{r},\varphi_\mathrm{r}) =\frac{1}{2}\left(\mathbf{F}_{\sigma \rho \nu}^{(3)}(r_\mathrm{r},\theta_\mathrm{r},\varphi_\mathrm{r}) + \mathbf{F}_{\sigma \rho \nu}^{(4)}(r_\mathrm{r},\theta_\mathrm{r},\varphi_\mathrm{r})\right),
  	\end{aligned}
\end{equation}
where $\mathbf{F}_{\sigma \rho \nu}^{(3)}$ and $\mathbf{F}_{\sigma \rho \nu}^{(4)}$ correspond to the outward- and inward-propagating components, respectively.
Noting that only the inward-propagating wave contributes to the response at the port of the receive antenna, the effective incident electric field $\mathbf{E}_{\mathrm{r}}$ at the receiver is given by
\begin{equation}\label{eq:Er_eff_in_Or}
	\begin{aligned}
  		\mathbf{E}_{\mathrm{r}} & (r_\mathrm{r},\theta_\mathrm{r},\varphi_\mathrm{r}) = 
  		\frac{k v_i}{2\sqrt{\eta}} \sum_{s=1}^{2} \sum_{n=1}^{N}\sum_{m=-n}^{n}  \sum_{\mu=-n}^{n}  \sum_{\gamma=-n}^{n} \sum_{\sigma=1}^{2}  \sum_{\substack{\nu=|\gamma| \\ \nu \neq 0}}^{N} \sum_{\rho=-\nu}^{\nu} \\
  		 & T_{smn} e^{\jmath [(\mu - \rho) \varphi_{ij} -\mu  \alpha_i]} d^n_{\mu m}(-\beta_{i})  d^{n}_{\gamma \mu}(\theta_{ij}) d^{\nu}_{\rho \gamma}(-\theta_{ij})   \\
  		 & C^{sn(c)}_{\sigma \gamma \nu}(k \Vert \mathbf{d}_{ij} \Vert) \mathbf{F}_{\sigma \rho \nu}^{(4)}(r_\mathrm{r},\theta_\mathrm{r},\varphi_\mathrm{r}).
  	\end{aligned}
\end{equation}
According to (\ref{eq:Er_receive_signal_form}), the received signal $w_j$ at the receive antenna's port can be expressed as
\begin{equation}\label{eq:receive_w_form}
	\begin{aligned}
  		w_j = &\frac{k v_i}{2\sqrt{\eta}} \sum_{s=1}^{2} \sum_{n=1}^{N}\sum_{m=-n}^{n}  \sum_{\mu=-n}^{n}  \sum_{\gamma=-n}^{n} \sum_{\sigma=1}^{2}  \sum_{\substack{\nu=|\gamma| \\ \nu \neq 0}}^{N} \sum_{\rho=-\nu}^{\nu}  \\
  		 &  T_{smn} R_{\sigma \rho \nu}  e^{\jmath [(\mu - \rho) \varphi_{ij} -\mu  \alpha_i]} 
  		d^n_{\mu m}(-\beta_{i})  d^{n}_{\gamma \mu}(\theta_{ij})   \\
  		 & d^{\nu}_{\rho \gamma}(-\theta_{ij})  C^{sn(c)}_{\sigma \gamma \nu}(k \Vert \mathbf{d}_{ij} \Vert) .
  	\end{aligned}
\end{equation}
By leveraging the index mapping $(s,m,n) \leftrightarrow p$ and $(\sigma,\rho,\nu) \leftrightarrow q$, 
the channel gain $h_{ij} = w_j / v_i$ between the transmit and receive antenna can be expressed in a matrix form as
\begin{equation}
	\begin{aligned}
  		h_{ij} = \mathbf{r}^T \tilde{{\mathbf{G}}}(\mathbf{n}_{i},\mathbf{d}_{ij}) \mathbf{t}.
  	\end{aligned}
\end{equation}
where the $p,q$-th element of ${\mathbf{G}}$ is given by
\begin{equation}
	\begin{aligned}\label{eq:gene_G}
  		[\tilde{{\mathbf{G}}}(\mathbf{n}_{i},& \mathbf{d}_{ij})]_{p,q} =  \frac{k}{2\sqrt{\eta}}  \sum_{\mu=-n}^{n}     \sum_{\gamma=-n}^{n} 
  		  e^{\jmath [(\mu - \rho) \varphi_{ij} -\mu  \alpha_i]} \\
  		 & d^n_{\mu m}(-\beta_{i})  d^{n}_{\gamma \mu}(\theta_{ij})  d^{\nu}_{\rho \gamma}(-\theta_{ij})  C^{sn(c)}_{\sigma \gamma \nu}(k \Vert \mathbf{d}_{ij} \Vert) 
  	\end{aligned}
\end{equation}
Under the far-field assumption, i.e., $k \Vert \mathbf{d}_{ij} \Vert \rightarrow \infty$, and using the asymptotic expression of the translation coefficient $C^{sn(3)}_{\sigma \mu \nu}(k \Vert \mathbf{d}_{ij} \Vert)$ \cite[eq.~(A3.22), (A3.23), (A3.23)]{Hansen_1988},
(\ref{eq:gene_G}) can be further simplified to
\begin{equation}
	\begin{aligned}\label{eq:far_field_G}
  		[\tilde{{\mathbf{G}}}(\mathbf{n}_{i},\mathbf{d}_{ij})]_{p,q} =  \frac{e^{\jmath k \Vert \mathbf{d}_{ij} \Vert}}{{k \Vert \mathbf{d}_{ij} \Vert}} [{{\mathbf{G}}}(\mathbf{n}_{i},\mathbf{d}_{ij})]_{p,q},
  	\end{aligned}
\end{equation}
where $[{{\mathbf{G}}}(\mathbf{n}_{i},\mathbf{d}_{ij})]_{p,q}$ is defined in (\ref{eq:G_pq}), which completes the proof.
\qed

\bibliographystyle{IEEEtran}
\bibliography{references}

\end{document}